\documentclass[letterpaper, 10 pt, conference]{ieeeconf} 
\IEEEoverridecommandlockouts                              
 
\usepackage[utf8]{inputenc}  
  
\usepackage{mathtools}
\mathtoolsset{showonlyrefs,showmanualtags}
\usepackage{amssymb,amsthm}
\usepackage{dsfont}
\usepackage{booktabs} 
\usepackage{color}
\usepackage{bbold} 
\usepackage[yyyymmdd,hhmmss]{datetime}
\usepackage{bm,upgreek}
\usepackage{tablefootnote}
\usepackage{nth}
\usepackage{pgfplots}
\usepackage{tikz}
\usetikzlibrary{shapes,snakes}
\usetikzlibrary{patterns}
\usepackage{multirow}
\usepackage{footnote}
\usetikzlibrary{hobby}

\usepackage{algorithm}
\usepackage{textcomp}
\usepackage{xcolor}
\usepackage{multicol}
\usepackage{algorithmic}
\def\BibTeX{{\rm B\kern-.05em{\sc i\kern-.025em b}\kern-.08em
    T\kern-.1667em\lower.7ex\hbox{E}\kern-.125emX}}

\newtheorem{assumption}{Assumption}
\newtheorem{theorem}{Theorem}

\newtheorem{remark}{Remark}
\newtheorem{lemma}{Lemma}

\newtheorem{problem}{Problem}

\DeclareMathOperator{\DIAG}{diag}
\DeclareMathOperator{\VEC}{vec}
\DeclareMathOperator{\Poly}{poly}

\DeclareMathOperator{\TR}{Tr}

\newcommand{\norm}[1]{ \| #1 \| }

\newcommand{\Exp}[1]{\mathbb{E}\big[ #1\big]}
       
\newcommand{\NF}[1]{\| #1\|_F}  
 
\newcommand{\barF}{{{}\bar F_n}}
\newcommand{\DeltaF}{{{} F_n}}
\newcommand{\barK}{{{}\bar K_n}}
\newcommand{\DeltaK}{{{} K_n}}

\newcommand{\Compress}{\medmuskip=0mu
\thinmuskip=0mu
\thickmuskip=0mu}

 \usepackage{tikz}
\usetikzlibrary{graphs,graphs.standard}

\tikzset{%
  every neuron/.style={
    circle,
    draw,
    minimum size=.7cm
  },
  neuron missing/.style={
    draw=none, 
    scale=3,
    text height=0.333cm,
    execute at begin node=\color{black}$\vdots$
  },
}

\definecolor{red}{RGB}{187,0,0}
\definecolor{blue}{RGB}{0, 0,180}
\definecolor{pink}{RGB}{203, 76, 178}

\title{\LARGE \bf Reinforcement Learning in Nonzero-sum  Linear Quadratic Deep Structured Games: Global Convergence of Policy  Optimization}

\author{Masoud Roudneshin, Jalal Arabneydi and Amir G. Aghdam
\thanks{This work is supported in part by the Natural Sciences and Engineering Research Council of Canada (NSERC) under Grant RGPIN-262127-17.}  
\thanks{Masoud Roudneshin, Jalal Arabneydi, and Amir G. Aghdam are with the  Department of Electrical and Computer Engineering, 
        Concordia University, 1455 de Maisonneuve Blvd. West, Montreal, QC, Canada, Postal Code: H3G 1M8.  Email: {\tt\small m\_roundne@encs.concordia.ca}, {\tt\small jalal.arabneydi@mail.mcgill.ca},        
        {\tt\small aghdam@ece.concordia.ca}}%
}

\begin{document}
\maketitle

\vspace*{-5cm}{\footnotesize{Proceedings of IEEE  Conference on Decision and Control, 2020.}}
\vspace*{3.85cm}

\thispagestyle{empty}
\pagestyle{empty}
\begin{abstract}
We  study  model-based and model-free policy optimization in  a class of nonzero-sum stochastic dynamic  games called  linear quadratic (LQ) deep structured  games. In such games,  players interact with each other through a set of weighted averages (linear regressions)  of  the states and actions. In this paper, we focus  our attention to homogeneous weights; however,  for the special case of infinite population, the obtained results extend to asymptotically vanishing weights wherein the players learn the  sequential weighted mean-field equilibrium.  Despite the non-convexity of the optimization in policy space and the fact that policy optimization  does not generally converge in game setting,   we prove  that the proposed  model-based and model-free policy gradient descent and natural policy gradient descent algorithms  globally converge to the  sub-game  perfect Nash equilibrium.    To the best of our knowledge, this is  the first  result  that provides a global convergence proof of policy optimization  in  a nonzero-sum  LQ game. One of the salient  features of the proposed  algorithms  is that  their parameter space  is independent of the number of players, and when  the dimension of  state space is significantly larger than that of the action space,  they  provide a more efficient way of computation compared to  those algorithms that plan and learn in the action space.  Finally, some simulations are  provided to numerically verify the obtained  theoretical results. 
\end{abstract}

%

\section{Introduction}

In recent years, there has been a growing interest in  the application of reinforcement learning (RL) algorithms  in networked control systems. One of  the most popular reinforcement learning (RL) algorithms in practice  is policy gradient, due to  its stability and fast convergence. However, from the theoretical point of view,  there is not much known about it.  Recently, it is shown in~\cite{fazel2018global} that a single-agent linear quadratic (LQ) optimal control  problem enjoys the global convergence, despite the fact that the optimization problem is not convex in the policy space. A similar result is obtained  for zero-sum LQ games in~\cite{zhang2019policy}. On the other hand, a nonzero-sum LQ game is  more challenging than the  above problems,
  where the existing results on the global (or even local) convergence of the policy gradient methods  are generally not encouraging~\cite{mazumdar2019policy}.  

Inspired by recent developments in deep structured teams and games~\cite{Jalal2019MFT,Jalal2019Automatica,Jalal2019risk,Jalal2020Nash,Jalal2020CCTA,Vida2020CDC},  we study a class of LQ games wherein  the  effect of other  players on any individual player is characterized by a linear regression of the states and actions of all players. The closest field of research to  deep structured games  is mean-field games~\cite{Caines2018book}. In a classical LQ mean-field game, one often has: (a)  homogeneous individual weights (i.e., players are equally important); (b)  the number of players~$n$ is asymptotically large with independent  primitive random variables (to be able to predict the trajectory of the mean-field using the strong law of large numbers); (c) the coupling is through the mean of the states, where the control coupling (called extended coupling) is  more challenging;  (d) the proof technique revolves  around the fact that the effect of a  single player on others is negligible,   reducing the game to a coupled  forward-backward  optimal control problem; (e) the solution concept is Nash equilibrium; (f)  given some fixed-point conditions across  the time horizon,  the forward-backward equation admits a solution leading to  an approximate Nash in the finite-population game;  (g) they are often not  practical for long-horizon  and reinforcement learning applications  wherein the  common practice  is to adopt a weaker solution  concept called stationary Nash equilibrium (where the trajectory of the mean-field is stationary), and (h)  since the results are  asymptotic,  the models are limited to those  that are  uniformly  bounded in~$n$. In contrast to mean-field game,  LQ deep structured game often has: (a')  heterogeneous individual weights that are not necessarily homogeneous; (b')  the number of players  is arbitrary (not necessarily very large) with possibly  correlated primitive random variables; (c')  the coupling is  through the weighted mean of the states and actions; (d') the proof technique  revolves around  a gauge transformation initially proposed in~\cite{arabneydi2016new} (not  based on the negligible effect); (e') the solution concept is sequential Nash; (f') the solution is  exact (not  an approximate one) for any arbitrary number of players  and it is  identified by Riccati equations; (g') since the solution concept is sequential,  it is well suited  for long-horizon and  reinforcement learning, and (h') since the results are also valid for  finite-population game, the dynamics and cost are not necessarily  limited to uniformly bounded functions with respect to $n$. It  is shown in~\cite{Jalal2019Automatica} that the classical LQ mean-field game with the tracking cost formulation is a special case of deep structured games under standard conditions, where  the mean-field equilibrium coincides with the sequential mean-field equilibrium. It is to be noted that the LQ mean-field-type game~\cite{elliott2013discrete,bensoussan2013mean,carmona2018probabilistic} is a single-agent
control problem (i.e., it is not a non-cooperative game), which resembles a team problem with social welfare cost function.\footnote{When the mean field  is replaced by the  expectation of the state of the genetic player, the resultant problem is called mean-field-type game.} In particular, it may be viewed as a special case of  risk-neutral LQ mean-field teams introduced in~\cite{arabneydi2016new},  showcased in~\cite{Jalal2017linear,JalalCDC2015,JalalCDC2018,Jalal2019LCSS,JalalACC2018},  and  extended to deep structured LQ teams in~\cite{Jalal2019risk}. 
The interested reader is referred to~\cite[Section VI]{Jalal2019Automatica} for more details on similarities and differences between mean-field games, mean-field-type games and  mean-field teams.

The rest of the paper is organized as follows. In Section~\ref{sec:problem}, the problem  of LQ deep structured game  is formulated. In Section~\ref{sec:main},  the global convergence of model-based and model-free policy gradient descent and natural policy gradient descent algorithms are presented. In Section~\ref{sec:numerical}, some numerical examples are provided to validate the theoretical results. The paper is concluded in Section~\ref{sec:conclusion}.

\section{Problem Formulation}\label{sec:problem}
Throughout the paper, $\mathbb{R}$, $\mathbb{R}_{>0}$ and $\mathbb{N}$ refer to the sets of real, positive real and natural numbers, respectively.  Given any $ n \in \mathbb{N}$, $\mathbb{N}_n$, $x_{1:n}$ and $\mathbf{I}_{n \times n}$ denote the finite set $\{1,\ldots,n\}$, vector $(x_1,\ldots,x_n)$ and the $n\times n$ identity matrix, respectively.    $\| \boldsymbol \cdot \|$ is the  spectral norm of a matrix,  $\NF{\boldsymbol \cdot}$ is the Frobenius norm of a matrix, $\TR(\boldsymbol \cdot)$ is the trace of a matrix,  $\sigma_{\text{min}}(\boldsymbol \cdot)$ is the minimum singular value of a matrix, $\rho(\boldsymbol \cdot)$ is  the spectral radius of a matrix,   and  $\DIAG(\Lambda_1, \Lambda_2)$ is the block diagonal matrix $[\Lambda_1\quad 0;0 \quad \Lambda_2]$. For vectors $x,y$ and $z$, $\VEC(x,y,z)=[x^\intercal, y^\intercal,z^\intercal]^\intercal$ is a column vector.   The superscript $-i$ refers to all players except the $i$-th player. In addition, $\Poly(\boldsymbol \cdot)$ denotes polynomial function.

Consider a nonzero-sum stochastic dynamic game with $n \in \mathbb{N}$  players. Let $x^i_t \in \mathbb{R}^{d_x}$, $u^i_t \in \mathbb{R}^{d_u}$ and $w^i_t \in \mathbb{R}^{d_x}$ denote the state, action and local noise of player $i \in \mathbb{N}_n$ at time $t \in \mathbb{N}$, where $d_x,d_u \in \mathbb{N}$.     Define  the  weighted averages:
\begin{equation}
\bar x_t:=\sum_{i=1}^n  \alpha^i_n x^i_t, \quad \bar u_t:=\sum_{i=1}^n \alpha^i_n u^i_t,
\end{equation}
where $\alpha^i_n \in \mathbb{R}$  is the \emph{influence factor} (weight) of player $i$ among its peers. From~\cite{Jalal2019MFT,Jalal2019Automatica,Jalal2019risk}, we refer to the above linear regressions as \emph{deep state} and \emph{deep action} in the sequel.   To ease the presentation,   the weights are normalized as follows:  $\sum_{i=1}^n \alpha^i_n=1$.

The initial states $\{x^1_1,\ldots,x^n_1 \}$ are  random  with   finite covariance matrices.   The evolution of  the state of player $i \in \mathbb{N}_n$ at time $t \in \mathbb{N}$  is given by:
\begin{equation}\label{eq:dynamics_original}
x^i_{t+1}=Ax^i_t +Bu^i_t+ \bar A \bar x_t +\bar B \bar u_t +w^i_t,
\end{equation}
where $\{w^i_t\}_{t=1}^\infty$ is an i.i.d. zero-mean  noise process with  a finite  covariance matrix.   The primitive  random variables $\{ \{x^i_1\}_{i=1}^n,  \{w^i_1\}_{i=1}^n,\{w^i_2\}_{i=1}^n,\ldots \}$ are defined on a common probability space and are mutually independent across time. The above random variables can be non-Gaussian  and  correlated (not necessarily independent) across players.  The cost of player $i \in \mathbb{N}_n$ at time $t \in \mathbb{N}$ is given by:
\begin{equation}\label{eq:cost_original}
\begin{split}
c_t^i=&(x_t^i)^{\intercal}Q x_t^i+2(x_t^i)^{\intercal}S^x\bar{x}_t+(\bar{x}_t)^{\intercal}\bar Q\bar{x}_t\\
&+(u_t^i)^{\intercal}R u_t^i+2(u_t^i)^{\intercal}S^u\bar{u}_t+(\bar{u}_t)^{\intercal} \bar R\bar{u}_t,
\end{split}
\end{equation}
where $Q, S^x, \bar{Q}, R, S^u$ and $\bar{R}$ are symmetric matrices with appropriate dimensions. 

From~\cite{Jalal2019MFT,Jalal2019Automatica,Jalal2019risk}, an information structure called \emph{deep state sharing} (DSS) is considered wherein each player $i \in \mathbb{N}_n$ at any time $t \in \mathbb{N}$ observes its local state $x^i_t$ and  the deep  state~$\bar x_t$, i.e., $u^i_t=g^i_t(x^i_{1:t},\bar x_{1:t})$,
where $g^i_t$ is a measurable function adapted to the filteration  of the underlying  primitive random variables of $\{x^i_{1:t},\bar x_{1:t}\}$.
 When the number of players is very large, one can use \emph{no-sharing} (NS) information structure wherein each player observes  only its local state.  However, such a  fully decentralized information structure comes at a price that one must predict the trajectory of the deep state in time (which introduces the computational complexity in time horizon in terms of storage and computation). For example,  if  the dynamics of the deep state  (i.e., $A+\bar A$ and $B+\bar B$) is known (which is not applicable for model-free applications), the deep state  can be predicted a head of time  when primitive random variables are mutually independent  by the strong law of large numbers.  Alternatively,  one can assume to have access to an external simulator for the dynamics of the deep state (which is basically DSS structure).  In this paper,  we focus on DSS information structure wherein there is no loss of optimality in restricting attention to  stationary strategies despite the fact that the deep state is not stationary. 
   The  interested  reader is  referred   to~\cite{Jalal2019Automatica} for the convergence analysis of NS (approximate) solution to  the  DSS  solution, as $n \rightarrow \infty$.

 Define  $\mathbf g^i_n:=\{g^i_t\}_{t=1}^\infty$ and $\mathbf g_n:=\{ \mathbf g^1,\ldots,\mathbf{g}^n\}$.  The admissible set of actions are square integrable such that $\Exp{\sum_{t=1}^\infty  \gamma^{t-1}(u^i_t)^\intercal u^i_t} <\infty$. Given a discount factor $\gamma \in (0,1)$, the cost-to-go for any  player $i \in \mathbb{N}_n$ is described by:
\begin{equation}\label{cost_function}
J^i_{n,\gamma} (\mathbf g^i_n,\mathbf g^{-i}_n)_{t_0}=(1-\gamma) \Exp{\sum_{t=t_0}^\infty \gamma^{t-1} c^i_t}, \quad t_0 \in \mathbb{N}.
\end{equation}

\begin{problem}\label{prob1}
Suppose that the weights are homogeneous, i.e. $\alpha^i_n=\frac{1}{n}$, $i \in \mathbb{N}_n$. When a  sequential  Nash strategy $\mathbf{g}^\ast_n $  exists,  develop model-based and model-free  gradient descent and natural policy gradient descent procedures  under DSS information structure such that  for any player $i \in \mathbb{N}_n$ at any stage of the game $t_0\in \mathbb{N}$, and any arbitrary strategy $\mathbf g^i$:
\begin{equation}
{J^i_{n,\gamma}(\mathbf g^{\ast,i}_n,\mathbf g^{\ast,-i}_n)}_{t_0} \leq {J^i_{n,\gamma}(\mathbf g^{ i},\mathbf g^{\ast,-i}_n)}_{t_0}. 
\end{equation}
\end{problem}

\begin{remark}
\emph{It is to be noted that Problem~\ref{prob1} holds  for  arbitrary number of players $n$, where the solution depends on $n$. Since the infinite-population solution is easier for analysis and may be viewed as a special case, one can generalize  the  homogeneous weights $\alpha^i_n=\frac{1}{n}$ to   heterogeneous weights   $\alpha^i_n=\frac{\beta^i}{n}$,  where $\beta^i \in [-\beta_{\text{max}},\beta_{\text{max}}]$, $\beta_{\text{max}} \in \mathbb{R}_{>0}$, $i \in \mathbb{N}_n$.     The resultant solution is called \emph{sequential weighted mean-field equilibrium} (SWMFE) in  \cite{Jalal2019Automatica}. The SWMFE constructs an approximate solution at any stage of the game $t_0 \in \mathbb{N}$ such that $J^i_{n,\gamma}(\mathbf g^{\ast,i}_\infty,\mathbf g^{\ast,-i}_{\infty})_{t_0} \leq {J^i_{n,\gamma}(\mathbf g^{ i},\mathbf g^{\ast,-i}_\infty)}_{t_0}+\varepsilon(n)$, 
where $\lim_{n \rightarrow \infty}\varepsilon(n)=0$. For more details, see~\cite[Theorem 4]{Jalal2019Automatica}.}
\end{remark}

\subsection{Main challenges and contributions}
There are several challenges to solve Problem~\ref{prob1}. The first one is the \emph{curse of dimensionality}, where  the computational complexity of the solution increases with the number of players.   The second one is the  \emph{imperfect information} structure, where players do not have perfect information about  the states of other players. The third challenge is that the resultant  optimization problem  is \emph{non-convex} in the policy space, see a counterexample in~\cite{fazel2018global}. The forth one lies in the fact that policy optimization is \emph{not even locally convergent} in a game with continuous spaces, in general; see a counterexample in~\cite{mazumdar2019policy}. The main contribution of this paper is  to  present an analytical proof for the global convergence of  model-based  and model-free policy gradient  algorithms. In contrast to the model-based solution in~\cite{Jalal2019Automatica} (whose number of unknowns increases quadratically with $d_x$), the number of unknown parameters in the proposed  algorithms  increases linearly with  $d_x$ and $ d_u$.   To the best of our knowledge, this is the first result on the global convergence of policy  optimization   in  nonzero-sum LQ games.

\section{Main Results}\label{sec:main}
In this section, we first present a model-based algorithm introduced in~\cite{Jalal2019Automatica} that requires $2d_x \times 2d_x$ parameters to construct the solution. Then, we propose two model-based gradient algorithms and prove their global convergence to the above solution, where their planning space is the  policy space (that requires $2d_u \times d_x$ parameters to identify the solution).   Based on the proposed  gradient methods, we  develop two model-free (reinforcement learning) algorithms and establish their global convergence to the model-based solution.

From~\cite{Jalal2019Automatica}, we use a gauge transformation to define the  following  variables for any player $i \in \mathbb{N}_n$ at any time $t \in \mathbb{N}$: 
$ \mathbf{x}_t^i:= \VEC(x_t^i-\bar{x}_t,\bar x_t)$, $ \mathbf{u}_t^i:= \VEC(u_t^i-\bar{u}_t,\bar u_t)$ and $ \mathbf{w}_t^i:= \VEC(w_t^i-\bar{w}_t,\bar w_t)$, 
where $\bar{w}_t:=\frac{1}{n}\sum_{i=1}^{n}  w_t^i$. In addition, we define  the following  matrices: $\mathbf A:=\DIAG(A, A+\bar A)$, $\mathbf B:=\DIAG(B, B+\bar B)$, and 
\begin{equation}
\Compress
\mathbf {Q}:=\begin{bmatrix}
Q  &Q+S^x\\
Q+S^x&Q+2S^x+\bar Q\\
\end{bmatrix},\quad 
\mathbf {R}:=\begin{bmatrix}
R&R+S^u\\
R+S^u&R+2S^u+ \bar R\\
\end{bmatrix}.
\end{equation}
We now express the per-step cost  of each player in~\eqref{eq:cost_original}  as:
\begin{align}\label{perstep}
c_t^i&=(\mathbf x^i_{t})^{\intercal}\mathbf{Q}\mathbf x^i_{t}+(\mathbf u^i_{t})^{\intercal}\mathbf{R}\mathbf u^i_{t}.
\end{align}
To  formulate the solution, we  present a non-standard algebraic Riccati equation, introduced in~\cite{Jalal2019Automatica}, as follows:
\begin{equation}\label{eq:riccati-bar-m}
\mathbf M(\boldsymbol \theta)=\mathbf Q + \boldsymbol \theta^\intercal \mathbf R \boldsymbol \theta+ \gamma (\mathbf A -\mathbf B \boldsymbol \theta)^\intercal \mathbf M(\boldsymbol \theta) (\mathbf A - \mathbf B \boldsymbol \theta), 
\end{equation}
where $\boldsymbol \theta:=\DIAG(\theta(n),\bar \theta(n))$,  $\theta(n):=(F_n)^{-1}  K_n$, $\bar \theta(n):=(\bar F_n)^{-1} \bar K_n$, and   matrices $\DeltaF$, $\barF$, $\DeltaK$ and $\barK$ are  given by:
\begin{align}\label{eq:breve-f}
&\DeltaF= (1-\frac{1}{n})\Big[R + \gamma B^\intercal {{}\mathbf M}^{ 1,1}(\boldsymbol \theta) B  \Big] \nonumber  \\
&\quad + \frac{1}{n}\Big[R + S^u +\gamma  (B+\bar B)^\intercal {{}\mathbf M}^{1,2}(\boldsymbol \theta) B \Big], \nonumber \\
&\barF= (1-\frac{1}{n})\left[R+ S^u +\gamma  B^\intercal {{}\mathbf M}^{2,1}(\boldsymbol \theta) (B+\bar B)  \right]  \nonumber \\
 &+ \frac{1}{n}\left[R + 2S^u + \bar R +\gamma  (B+\bar B)^\intercal {{}\mathbf M}^{2,2}(\boldsymbol \theta) (B + \bar B)  \right], \nonumber \\
&\DeltaK= (1-\frac{1}{n})\gamma B^\intercal {{}\mathbf M}^{1,1}(\boldsymbol \theta) A + \frac{\gamma}{n} (B+ \bar B)^\intercal {{}\mathbf M}^{1,2}(\boldsymbol \theta) A, \nonumber \\
&\bar  K_n\hspace{-.1cm}= \hspace{-.1cm}(1\hspace{-.1cm}-\hspace{-.1cm}\frac{1}{n})\gamma  B^\intercal {{}\mathbf M}^{ 2,1}(\boldsymbol \theta) (A\hspace{-.1cm}+\hspace{-.1cm}\bar A)\hspace{-.1cm} + \hspace{-.1cm} \frac{\gamma}{n}  (B\hspace{-.1cm}+ \hspace{-.1cm}\bar B)^\intercal {{}\mathbf M}^{2,2}(\boldsymbol \theta) (A\hspace{-.1cm}+\hspace{-.1cm}\bar A).
\end{align}
\begin{assumption}\label{ass:existence}
Suppose equations~\eqref{eq:riccati-bar-m} and~\eqref{eq:breve-f} admit a  unique stable solution, which is also the limit of the finite-horizon solution. In addition,  let $\DeltaF$ and $\barF$ be  invertible matrices, and  $(1-\frac{1}{n})\DeltaF + \frac{1}{n}\barF$  be  a positive definite matrix.
\end{assumption}
We now provide two sufficient conditions  for Assumption~\ref{ass:existence} ensuring the existence of a stationary solution. Let  $G$ denote the mapping  from  $\mathbf M$ to $\boldsymbol \theta$ displayed in~\eqref{eq:breve-f} (where $\boldsymbol \theta=G(\mathbf M)$),  and $L$ denote the mapping from $\boldsymbol \theta$ to $\mathbf M$ expressed in~\eqref{eq:riccati-bar-m} (where $\mathbf M=L(\boldsymbol \theta)$). Thus, $\mathbf M=L(G(\mathbf M))$ is a fixed-point equation to  be solved by  fixed-point methods.
\begin{assumption}\label{ass:contractive}
Let the mapping $L(G(\boldsymbol \cdot))$ be a contraction,   implying that equations~\eqref{eq:riccati-bar-m} and~\eqref{eq:breve-f} admit a unique  fixed-point  solution. In addition,  let $\DeltaF$ and $\barF$ be  invertible matrices, and  $(1-\frac{1}{n})\DeltaF + \frac{1}{n}\barF$  be  a positive definite matrix.
\end{assumption}
\begin{assumption}[Infinite-population decoupled  Riccati equations]\label{ass:decoupled}
\emph{Let $Q$ and $Q+S^x$ be positive semi-definite,  $R$ and $R+S^u$ be positive definite,  and $\bar A$ and  $\bar B$ be zero. Suppose  $(A,B)$ is stabilizable, and $(A,Q^{1/2})$ and $(A,(Q+S^x)^{1/2})$ are detectable.   When  $n$ is asymptotically large,  the non-standard Riccati equation~\eqref{eq:riccati-bar-m} decomposes into two decoupled standard Riccati equations; see~\cite[Proposition 2]{Jalal2019Automatica}.}
\end{assumption}
\begin{theorem}[Model-based solution using non-standard Riccati equation~\cite{Jalal2019Automatica}]\label{thm:model_known}
Let Assumption~\ref{ass:existence} hold.   There exists a stationary subgame  perfect Nash equilibrium such that  for any player $ i \in \mathbb{N}_n$ at any time  $t\in \mathbb{N}_T$,
 \begin{equation}
 u^{\ast,i}_t=-\theta^\ast(n) x^i_t -  ( \bar \theta^\ast(n) - \theta^\ast(n))\bar x_t,
  \end{equation}
  where the  gains are obtained from \eqref{eq:breve-f}. 
 In addition,  the optimal cost of player~$i \in \mathbb{N}_n$  from the initial time $t_0=1$ is given  by:
$ J^{i}_{n,\gamma}(\boldsymbol \theta^\ast)=    (1-\gamma) \TR( \mathbf M(\boldsymbol \theta^\ast) \boldsymbol \Sigma^i_x) +\gamma  \TR( \mathbf M(\boldsymbol \theta^\ast)  \boldsymbol \Sigma^i_w)$,
where $ \boldsymbol \Sigma^i_x:=\Exp{(\VEC(\Delta x^i_1), \bar x_1)(\VEC(\Delta x^i_1), \bar x_1)^\intercal}$ and $ \boldsymbol \Sigma^i_w:=\Exp{\VEC(\Delta w^i_t), \bar w_t)\VEC(\Delta w^i_t), \bar w_t)^\intercal }$.
\end{theorem}

\subsection{Model-based solution using policy optimization}
  From Theorem~\ref{thm:model_known}, there is no loss of optimality in restricting attention to  linear identical stationary strategies of the form $\boldsymbol \theta=\DIAG(\theta,\bar \theta)$.   Therefore,  we   select one arbitrary player~$i$  as a learner  and  other players as imitators  (that are passive during the learning process). 
   More precisely, at each time instant, player $i$ uses a gradient algorithm to update its strategy  whereas  other players employ  the updated strategy to  determine their next actions.  In this article, we discard the process of selecting the learner, but in order to have a fair implementation, the learner may be chosen randomly at each iteration.\footnote{For the special case of infinite population, it is also possible that all players become learners, i.e., they simultaneously learn the strategies as long as their  exploration noises are i.i.d. In such a case,  the infinite-population deep state  reduces to  weighted  mean-field and remains unchanged.
} For simplicity of presentation, we omit the superscript $i$ and the subscription  of the cost function.  Hence,  the strategy of the  learner  can be described by:
$\mathbf u_t=-\boldsymbol \theta \mathbf x_t,  \mathbf u_t \in \mathbb{R}^{2d_u}, \mathbf x_t \in \mathbb{R}^{2 d_x},  t \in \mathbb{N}$.

\begin{lemma}\label{lemma:gradient_formulation} The  following holds at the initial time $t_0=1$:
\begin{equation}\label{eq:gradient_1}
[
\nabla_{ \theta} J(\boldsymbol \theta),
\nabla_{\bar \theta} J(\boldsymbol  \theta)]
= 2   \mathbf P_n \mathbf E_{\boldsymbol \theta}  \boldsymbol \Sigma_{\boldsymbol \theta},
\end{equation}
where 
\begin{equation}\label{eq:gradient_matrices}
\begin{cases}
\mathbf P_n:= [
(1-\frac{1}{n}) \mathbf I_{d_u \times d_u},
 \frac{1}{n} \mathbf I_{d_u \times d_u}],\\
 \mathbf E_{\boldsymbol \theta}:=    (\mathbf R + \gamma \mathbf B^\intercal \mathbf M(\boldsymbol \theta) \mathbf B)\boldsymbol \theta
 - \gamma \mathbf B^\intercal \mathbf M(\boldsymbol \theta) \mathbf A,\\
 \boldsymbol \Sigma_{\boldsymbol \theta}:= \Exp{(1-\gamma)\sum_{t=1}^\infty \gamma^{t-1}\mathbf x_t \mathbf x_t^\intercal}.
 \end{cases}
\end{equation}
\end{lemma}
\begin{proof}
The proof is presented in Appendix~\ref{sec:proof_lemma:gradient_formulation}.
\end{proof}
In this paper, we consider two gradient-based methods.
\begin{itemize}
\item \textbf{Policy gradient descent}: 
\begin{equation}\label{eq:GD}
\boldsymbol \theta_{k+1}=\boldsymbol \theta_k - \eta \DIAG(
\nabla_{ \theta} J( \boldsymbol \theta_k),\nabla_{\bar \theta} J(\boldsymbol  \theta_k)).
\end{equation}
\item \textbf{Natural policy gradient descent}:
\begin{equation}\label{eq:NPGD}
\Compress
\boldsymbol \theta_{k+1}=\boldsymbol \theta_k - \eta \DIAG(
\nabla_{ \theta} J( \boldsymbol \theta_k),\nabla_{\bar \theta} J(\boldsymbol  \theta_k)) \boldsymbol \Sigma_{\boldsymbol \theta}^{-1}.
\end{equation}
\end{itemize}

To prove our convergence results, we impose extra  standard assumptions, described below.
\begin{assumption}\label{stable_search}
The initial policy is stable. A policy $\boldsymbol \theta$ is said to be stable  if $\rho(\mathbf A-  \mathbf B \boldsymbol \theta) <1 $. 
\end{assumption}

\begin{assumption}\label{ass:positive_sigma}
Given the learner, $\Exp{\mathbf x_1 (\mathbf x_1)^\intercal}$ is positive definite.  For the special case of i.i.d. initial states,   $\Exp{\mathbf x^i_1 (\mathbf x^i_1)^\intercal}$= $\DIAG((1-\frac{1}{n}) \text{cov}(x_1), \frac{1}{n}\text{cov}(x_1)+\Exp{x_1} \Exp{x_1}^\intercal)$ is  positive definite  if   $\text{cov}(x^i_1)=:\text{cov}(x_1)$ and  $\Exp{x^i_1}\Exp{x^i_1}^\intercal=: \Exp{x_1}\Exp{x_1}^\intercal$, $ i \in \mathbb{N}_n$, are  positive definite. 
\end{assumption}
\begin{assumption}\label{ass:positive_Q}
For finite-population model,  $\mathbf Q$ and $\mathbf R$ are positive definite matrices.  For the infinite-population case  satisfying Assumption~\ref{ass:decoupled},  $Q$ and  $Q+S$  are positive definite.
\end{assumption}

Assumptions~\ref{stable_search}--\ref{ass:positive_Q} are  standard conditions in the literature of LQ reinforcement learning~\cite{fazel2018global,malik2020derivative}, which ensure  that for any stable $\boldsymbol \theta$, $J(\boldsymbol \theta)$ is properly bounded and $\boldsymbol \Sigma_{\boldsymbol \theta} \succcurlyeq \Exp{\mathbf x_1 (\mathbf x_1)^\intercal}$ is positive definite. We now show that the best-response optimization at the learner  satisfies the Polyak-Lojasiewicz (PL) condition~\cite{Polyak1964, Lojasiewicz1963}, which is a relaxation of the notion of strong convexity. Let $\mu:=\sigma_{min}(\Exp{\mathbf x_1 \mathbf x_1^\intercal})$.
\begin{lemma} [PL condition]	\label{lemma:GD}
Let Assumptions~\ref{ass:existence},~\ref{stable_search},~\ref{ass:positive_sigma} and~\ref{ass:positive_Q} hold. Let also $\boldsymbol \theta^*$ be the  Nash policy in Theorem~\ref{thm:model_known}.   There exists  a positive constant $L_1 (\boldsymbol \theta^\ast)$ such that
	\begin{equation}
	 J(\boldsymbol \theta)- J(\boldsymbol \theta^*) \leq  L_1 (\boldsymbol \theta^\ast)
	 \NF{[\nabla_{\theta} J(\boldsymbol \theta), \nabla_{ \bar \theta} J( \boldsymbol \theta)]}^2,
	\end{equation} 
	where $L_1(\boldsymbol \theta^\ast)= \frac{n^2\norm{\boldsymbol \Sigma_{\boldsymbol \theta^*}}}{4\mu^2\sigma_{\text{min}}(\mathbf{R})}$. For the special case of infinite population  (i.e. $n=\infty$) with i.i.d. initial states under Assumptions~\ref{ass:decoupled},~\ref{stable_search},~\ref{ass:positive_sigma} and~\ref{ass:positive_Q}, one has $L_1(\boldsymbol \theta^\ast)= \frac{\norm{\boldsymbol \Sigma^{1,1}_{\boldsymbol \theta^*}}}{4\sigma_{min}(\text{cov}(x_1))^2\sigma_{\text{min}}(R)} +  \frac{\norm{\boldsymbol \Sigma^{2,2}_{\boldsymbol \theta^*}}}{4\sigma_{min}(\mathbb{E}[x_1]\mathbb{E}[x_1]^\intercal)^2\sigma_{\text{min}}(R+S^u)}$.
\end{lemma}
\begin{proof}
	The proof is presented in Appendix~\ref{sec:proof_lemma:GD}.
\end{proof} 

In  the following lemmas, we  show that the cost function and  its gradient are locally Lipschitz functions.
\begin{lemma}[Locally Lipschitz cost function]\label{CLip}
	For any $\boldsymbol \theta '$ satisfying the inequality $\NF{\boldsymbol \theta' - \boldsymbol \theta} < \varepsilon(\boldsymbol \theta)$,  there exists a positive constant $L_2(\boldsymbol \theta)$ such that  
$	|J(\boldsymbol \theta ')- J(\boldsymbol \theta )|\leq L_2(\boldsymbol \theta )\NF{\boldsymbol \theta'-\boldsymbol \theta}$,
	where the explicit expressions of  $\varepsilon(\boldsymbol \theta)$ and $L_2(\boldsymbol \theta)$ can be obtained in a similar manner as~\cite[Lemma  15]{malik2020derivative}.
\end{lemma}
\begin{proof}
 The proof is omitted due to space limitation. 
\end{proof}

\begin{lemma}[Locally Lipschitz gradient]	\label{lemma:LLG}
	For any $\boldsymbol \theta '$ satisfying the inequality $\NF{\boldsymbol \theta' - \boldsymbol \theta} < \varepsilon(\boldsymbol \theta)$, there exists a  positive constant  $L_3(\boldsymbol \theta )$  such that
	\begin{equation}
	\Compress
	\NF{[\nabla J_{\theta}(\boldsymbol \theta'), \nabla J_{\bar \theta}(\boldsymbol \theta')] - [\nabla J_{\theta}(\boldsymbol \theta), \nabla J_{\bar \theta}(\boldsymbol \theta)]}\leq L_3(\boldsymbol \theta)\NF{\boldsymbol \theta'-\boldsymbol \theta},
	\end{equation}
		where the explicit expressions of  $\varepsilon(\boldsymbol \theta)$ and $L_3(\boldsymbol \theta)$ can be obtained in a similar manner as~\cite[Lemma  16]{malik2020derivative}.
\end{lemma}
\begin{proof}
The proof is omitted due to space limitation. 
\end{proof}

\begin{theorem}[Global convergence via model-based gradient]\label{MDTHm}
Let Assumptions~\ref{ass:existence},~\ref{stable_search},~\ref{ass:positive_sigma} and~\ref{ass:positive_Q} hold.  For a sufficiently small  fixed step size $\eta$ chosen as 
$\eta=\Poly\big(\frac{\mu\sigma_{\text{min}}(\mathbf{Q})}{J (\boldsymbol \theta_1)},\frac{1}{\sqrt{\gamma}\rVert\mathbf{A}\lVert},\frac{1}{\sqrt{\gamma}\rVert\mathbf{B}\lVert},\frac{1}{\rVert\mathbf{R}\lVert},\sigma_{\text{min}}(\mathbf{R})\big)$,
and for a sufficiently large  number of iterations $K$ such that
$K \geq\frac{\rVert\mathbf{\Sigma_{\boldsymbol \theta^*}}\lVert}{\mu}\log\frac{J (\boldsymbol \theta_1)-J (\boldsymbol \theta^*)}{\varepsilon} \Poly\big(\frac{J (\boldsymbol \theta_1)}{\mu\sigma_{\text{min}}(\mathbf{Q})},{\sqrt{\gamma}\rVert\mathbf{A}\lVert},{\sqrt{\gamma}\rVert\mathbf{B}\lVert},{\rVert\mathbf{R}\lVert}$, $\frac{1}{\sigma_{\text{min}}(\mathbf{R})}\big)$,
the gradient descent algorithm~\eqref{eq:GD}  leads to  the following   bound:
$
J (\boldsymbol \theta_K)-J (\boldsymbol \theta^*)\leq\varepsilon
$.
In particular, for a fixed  step size $
\eta=\frac{1}{\|\mathbf{P}_n^\intercal \mathbf{P}_n\|(\rVert\mathbf{R}\lVert+\frac{\gamma\rVert\mathbf{B}\lVert^2J (\boldsymbol \theta_1)}{\mu})}$ 
and for a sufficiently  large number of iterations $K$, i.e.,
$K\geq\frac{\rVert\mathbf{\Sigma_{\boldsymbol \theta^*}}\lVert \|\mathbf P_n^\intercal \mathbf P_n\| }{\mu}\big(\frac{\rVert\mathbf{R}\lVert}{\sigma_{\text{min}}(\mathbf{R})}+\frac{\gamma\rVert\mathbf{B}\lVert^2J (\boldsymbol \theta_1)}{\mu\sigma_{\text{min}}(\mathbf{R})}\big)\log\frac{J (\boldsymbol \theta_1)-J (\boldsymbol \theta^*)}{\varepsilon}$,
the natural policy gradient descent algorithm~\eqref{eq:NPGD} enjoys the bound: $
J (\boldsymbol \theta_K)-J (\boldsymbol \theta^*)\leq\varepsilon$.
\end{theorem}
\begin{proof}
Following the proof technique  in~\cite[Theorem 7]{fazel2018global}, we choose a  sufficiently  small  step size $\eta$ such  that the value of the cost decreases at each iteration.  More precisely,  for the natural policy gradient descent at iteration $K$, 
	$J ({\boldsymbol \theta}_{K+1})-J (\boldsymbol \theta^*)\leq
	(1- \frac{\mu  \sigma_{\text{min}}(\mathbf{R})}{\|\mathbf{P}_n^\intercal \mathbf{P}_n\|(\rVert\mathbf{R}\lVert+\frac{\gamma\rVert\mathbf{B}\lVert^2J(\boldsymbol{ \theta_1})}{\mu})\rVert \mathbf{\Sigma_{\boldsymbol \theta^*}}\lVert})(J ({\boldsymbol \theta}_K)-J (\boldsymbol \theta^*)) = (1-\eta \frac{\mu  \sigma_{\text{min}}(\mathbf{R})}{\rVert \mathbf{\Sigma_{\boldsymbol \theta^*}}\lVert})(J ({\boldsymbol \theta}_K)-J (\boldsymbol \theta^*))$.  The above recursion is contractive  for the specified~$\eta$.
\end{proof}

%

\subsection{Model-free solution using policy optimization}
It is desired now to develop a model-free RL algorithm.

\begin{lemma}[Finite-horizon approximation]\label{lemma:rollout}
	For any $\boldsymbol \theta$ with finite $J(\boldsymbol \theta)$, define  ${\tilde J_T}(\boldsymbol \theta):=(1-\gamma)\mathbb{E}[\sum_{t=1}^{T} \gamma^{t-1} c_t]$ and $\tilde{\boldsymbol \Sigma}_{\boldsymbol \theta}=(1-\gamma)\mathbb{E}[\sum_{t=1}^{T} \gamma^{t-1} \mathbf x_t \mathbf x_t^\intercal]$.   Let $\varepsilon(T):= \frac{d_x(J (\boldsymbol \theta))^2}{(1-\gamma)T\mu\sigma_{\text{min}}^2(\mathbf{Q})}
	$ and 
	$
\bar \varepsilon(T):=\varepsilon(T) (\rVert\mathbf{Q}\lVert+\rVert\mathbf{R}\lVert\rVert\mathbf{\boldsymbol{ \theta}}\lVert^2)
	$, 
	then $\norm{\tilde{\boldsymbol \Sigma}_{\boldsymbol \theta}-\boldsymbol \Sigma_{\boldsymbol \theta}} \leq \varepsilon(T) $ and $|{\tilde J _T}(\boldsymbol \theta)-J (\boldsymbol \theta)|\leq\bar \varepsilon(T)$.
\end{lemma}
\begin{proof}
The proof is omitted due to space limitation.
\end{proof}

Let $\mathbb{S}_r$ be a set of uniformly distributed points with norm $r>0$ (e.g., the surface of a sphere). In addition, let $\mathbb{B}_r$ denote the set of all uniformly distributed points  whose norms are at most $r$ (e.g., all points within the sphere). For a matrix $\tilde{\boldsymbol \theta}=\DIAG(\tilde \theta, \tilde{\bar \theta})$, these distributions are defined over the Frobenius norm ball.  Hence,
$J_r (\boldsymbol \theta)=\mathbb{E}_{\tilde{\boldsymbol \theta} \sim \mathbb{B}_r}[J (\boldsymbol \theta+\tilde{\boldsymbol \theta})]$.
 Since the expectation can be expressed as an integral function, one can use  Stokes' formula  to  compute   the gradient of $J_r (\boldsymbol \theta )$ with only query access to the function values.

\begin{lemma}[Zeroth-order optimization]\label{lemma:smooth}
For a smoothing factor $r >0$, 
$
[\nabla_{ \theta} J ( \boldsymbol \theta),
\nabla_{\bar \theta} J (\boldsymbol  \theta)
]
= \frac{2d_xd_u}{r^2}  \mathbb{E}_{\tilde{\boldsymbol \theta} \sim \mathbb{S}_r }[J (\boldsymbol \theta+\tilde{\boldsymbol \theta})[\tilde \theta, \tilde{\bar \theta}]]$.
\end{lemma}
\begin{proof}
The proof follows directly from the zeroth-order optimization approach~\cite[Lemma 1]{flaxman2004online}.
%
\end{proof}

\begin{lemma}\label{lemma:ber}
	Let $\tilde {\boldsymbol \theta}_1, \ldots, \tilde {\boldsymbol \theta}_L$, $L \in \mathbb{N}$, be   i.i.d. samples drawn uniformly  from $\mathbb{S}_r$. There exists  $\varepsilon(L):= \Poly(1/L) >0$,  such that  $[
\tilde \nabla^L_{ \theta} J ( \boldsymbol \theta),
\tilde \nabla^L_{\bar \theta} J (\boldsymbol  \theta)
]
= \frac{2d_xd_u}{r^2 L}   \sum_{l=1}^{L}J (\boldsymbol \theta+\tilde{\boldsymbol \theta}_l)[\tilde \theta, \tilde{\bar \theta}]$  converges  to $[\nabla_{ \theta} J ( \boldsymbol \theta),\nabla_{\bar \theta} J (\boldsymbol  \theta)]$ in the Frobenius norm  with a probability greater than $1-({\frac{2d_xd_u}{\varepsilon(L)}})^{-2d_xd_u}$.
From Lemma~\ref{lemma:rollout},  there exists $\varepsilon(L,T):=\Poly(1/L, 1/T)>0$ such that 
$[
\tilde \nabla^{L,T}_{ \theta} J ( \boldsymbol \theta),
\tilde \nabla^{L,T}_{\bar \theta} J (\boldsymbol  \theta)
]
= \frac{2d_xd_u(1-\gamma)}{r^2 L} \sum_{l=1}^{L}[\sum_{t=1}^{T}\gamma^{t-1}(c_t)][\tilde \theta_l, \tilde{\bar \theta}_l]$
is $\varepsilon(L,T)$ close to $[\nabla_{ \theta} J ( \boldsymbol \theta),\nabla_{\bar \theta} J (\boldsymbol  \theta)]$   with a probability greater than $1-({\frac{2d_xd_u}{\varepsilon(L,T)}})^{-2d_xd_u}$ in the Frobenius norm.
	\end{lemma}
\begin{proof}
	The proof is omitted due to space limitation.
\end{proof}

\begin{theorem}[Global convergence via model-free gradient]\label{thm:RL}
Let Assumptions~\ref{ass:existence},~\ref{stable_search},~\ref{ass:positive_sigma} and~\ref{ass:positive_Q} hold. For a sufficiently large horizon $T$ and samples $L$,  model-free gradient descent and natural policy gradient decent with the empirical gradient in Lemma~\ref{lemma:ber} and covariance matrix  in Lemma~\ref{lemma:rollout} converge to the model-based solutions  in Theorem~\ref{MDTHm}. In particular, the gradient descent algorithm converges with a probability greater than $1-({\frac{2d_xd_u}{\varepsilon(L,T)}})^{-2d_xd_u}$,  where $\varepsilon(L,T)=\Poly(1/L, 1/T)$.
\end{theorem}
\begin{proof}
From~\cite[Theorem 31]{fazel2018global} and Theorem~\ref{MDTHm},  one has the following inequality  at iteration $K\in \mathbb{N}$  for a sufficiently small step size $\eta \leq \eta_{max}$,  
$J (\boldsymbol \theta_{K+1})-J (\boldsymbol \theta^*)\leq (1-\eta\eta_{max}^{-1})(J (\boldsymbol \theta_K)-J (\boldsymbol \theta^*))$.
At iteration $K$, denote by $\tilde{\nabla}_K$  the empirical  gradient and by $\hat{\boldsymbol{\theta}}_{K+1}=\boldsymbol{\theta}_K-\eta \tilde{\nabla}_K$ the update with  the empirical  gradient.  From Lemma~\ref{CLip},   $|J (\hat{\boldsymbol \theta}_{K+1})-J (\boldsymbol \theta_{K+1})| \leq  \frac{1}{2}\eta \eta_{max}^{-1}\varepsilon(L,T)$, when $\rVert\hat{\boldsymbol \theta}_{K+1}-\boldsymbol{ \theta}_{K+1} \lVert\leq \frac{1}{2}\eta \eta_{max}^{-1}\varepsilon(L,T) (1/L_2(\boldsymbol{ \theta}_{K+1}))$,  upon noting that  $\hat{\boldsymbol{ \theta}}_{K+1}-\boldsymbol{ \theta}_{K+1}=\eta(\nabla_K-\tilde{\nabla}_K)$ and
$\rVert \nabla_{K}-\tilde{\nabla}_{K}\lVert\leq \frac{1}{2} \eta_{max}^{-1}\varepsilon(L,T) (1/L_2(\boldsymbol{ \theta}_{K+1}))$.
According to the Bernstein inequality,  the above inequality holds with  a probability greater than $1-({\frac{2d_xd_u}{\varepsilon(L,T)}})^{-2d_xd_u}$. Therefore, from Lemmas~\ref{lemma:LLG} and~\ref{lemma:ber}, the distance between the empirical  gradient and  the exact one monotonically decreases as the number of samples and rollouts increases, provided that the smoothing factor $r$ is sufficiently small.  Consequently, one arrives at
$J (\hat{\boldsymbol \theta}_{K+1})-J (\boldsymbol \theta^*)\leq(1-\frac{1}{2}\eta \eta_{max}^{-1})(J (\boldsymbol \theta_K)-J (\boldsymbol \theta^*))$, when $ J (\boldsymbol \theta_K)-J (\boldsymbol \theta^*) \leq \varepsilon(L,T)$. This recursion is contractive; i.e.,  the rest of the proof will be similar to that of Theorem~\ref{MDTHm}. 
\end{proof}

\section{Simulations}\label{sec:numerical}
In this section, simulations are conducted to  demonstrate the  global convergence of the proposed gradient methods. To compute the Nash policy,  plotted in dashed lines  in the figures, we use  the solution of equation~\eqref{eq:riccati-bar-m}.

\textbf{Example 1.} Consider a dynamic game with the following parameters:  $\eta=0.1$, $n=100$, $T=100$, $L=3$, $A=0.7, B=0.4, \bar{A}=0, \bar{B}=0, Q=1, R=1,S^x=4, S^u=0, \bar{Q}=0$,  $\bar{R}=0$, $\Sigma_x=1$ and $\Sigma_w=0.4$.  It is observed in Figure~\ref{fig:PN} that natural  policy gradient descent reaches  the  Nash strategy  faster than   the gradient descent. 
\begin{figure}[t!]
	\hspace{0cm}
	\scalebox{1}{
	\includegraphics[ trim={0cm 13.6cm 0 7cm},clip,width=\linewidth]{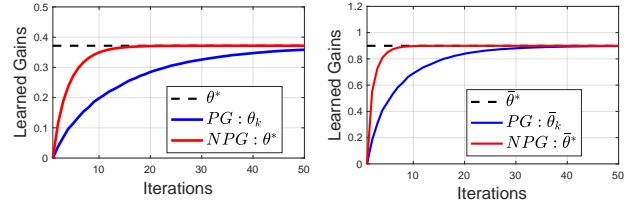}}
	\caption{Convergence of the model-based gradient descent and natural policy gradient descent algorithms in Example 1.}\label{fig:PN}
\end{figure}

\textbf{Example 2.} Let the system parameters be $\eta=0.04$, $n=10$, $T=10$, $r=0.09$, $L=1500$, $A=1, B=0.5, \bar{A}=0, \bar{B}=0, Q=1, R=1,S^x=2, S^u=0, \bar{Q}=1$, $\bar{R}=0$,  $r=0.09$, $\Sigma_x=0.05$ and $\Sigma_w=0.01$. The model-free policy gradient algorithm  was run on a 2.7 GHz Intel Core  i5 processor for $10$  random seeds.  After $6000$ iterations, which took roughly $10$ hours,  both $\theta$ and $\bar{ \theta}$ reached their optimal values   as depicted in Figure~\ref{fig:model_free}. 
\begin{figure}[t!]
	\centering 
	\includegraphics[trim={0cm 8.5cm 0 8.7cm},clip,width=\linewidth]{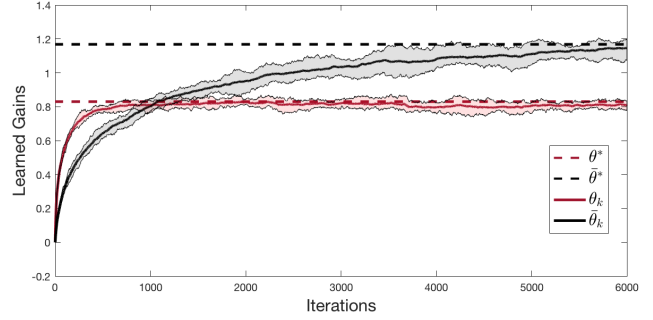}
	\caption{Convergence of the proposed model-free algorithm in Example 2.}\label{fig:model_free}
\end{figure}

 \begin{figure}[t!]
\centering
\includegraphics[ trim={0cm 9.3cm 0 9.4cm},clip,width=\linewidth]{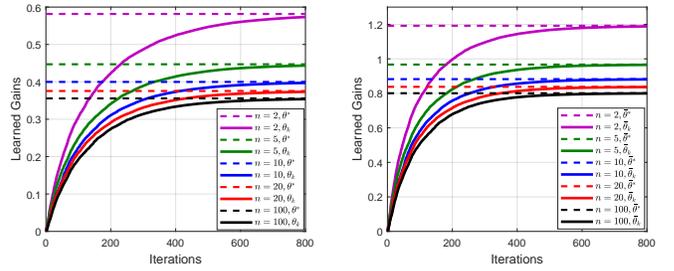}
\caption{The effect of the number of players on the policy~in Example~3.}\label{fig:bar}
\end{figure}
 
\textbf{Example 3.}  In this example, let the system parameters be $\eta=0.1$,  $T=100$, $L=3$, $A=0.8, B=0.2, \bar{A}=0, \bar{B}=0, Q=1, R=1,S^x=2, S^u=0, \bar{Q}=4$, $\bar{R}=0$,  $\Sigma_x=1$ and $\Sigma_w=0.1$. To investigate the effect of the number of players,  we  considered five different values for $n \in \{ 2, 5, 10,20,100\}$. It is shown in Figure~\ref{fig:bar} that the policies converge to a limit as the number of players increases, which is known as the mean-field limit. 

\section{Conclusions}\label{sec:conclusion}
In this paper, we  investigated  model-based and model-free gradient descent and natural policy gradient descent algorithms for LQ deep structured  games with  homogeneous weights.   It was shown theoretically and  verified by simulations, that the gradient-based methods enjoy the global convergence to the sequential Nash solution. 
  One of the main  features of the proposed  solutions  is that  their planning space  is independent of  the number of players. 
  The obtained results naturally extend to asymptotically vanishing weights and other variants of policy gradient algorithms such as REINFORCE and actor-critic methods.

\bibliographystyle{IEEEtran} 
\bibliography{Jalal_Ref}

\appendices
\section{Proof of Lemma~\ref{lemma:gradient_formulation}}\label{sec:proof_lemma:gradient_formulation}
To compute the best-response of the learner,  we fix the strategies of  other players, and then find the gradient of the cost function with respect to $\theta$ and $\bar \theta$. Suppose  player $i\in \mathbb{N}_n$ uses the strategy $u^i_t=\theta^i x^i_t+(\bar \theta^i -\theta^i) \bar x_t$. Therefore, one has:
\begin{equation}
\mathbf u^i_t= \begin{bmatrix}
(1-\frac{1}{n}) \theta^i & (1-\frac{1}{n}) \bar \theta^i\\
\frac{1}{n} \theta^i & \frac{1}{n} \bar{\theta}^i
\end{bmatrix}
\mathbf x^i_t
+ \sum_{j \neq i}
 \begin{bmatrix}
-\frac{1}{n} \theta^j & -\frac{1}{n} \bar \theta^j\\
\frac{1}{n} \theta^j & \frac{1}{n} \bar{\theta}^j
\end{bmatrix}
\mathbf x^j_t.
\end{equation}
From~\eqref{cost_function} and \eqref{perstep}, 
$ J_{\mathbf x^i_1}( \boldsymbol \theta)=\Exp{(\mathbf x^i_1)^{\intercal}\mathbf Q \mathbf x^i_1
+(\mathbf u^i_1)^{\intercal}\mathbf R \mathbf u^i_1}+  \gamma  J_{ \mathbf x^i_2}( \boldsymbol \theta)$ 
$=\Exp{(\mathbf x^i_1)^{\intercal}\mathbf Q \mathbf x^i_1
+(\mathbf u^i_1)^{\intercal}\mathbf R \mathbf u^i_1} +  \gamma  \Exp{(\mathbf x^i_2)^{\intercal} \mathbf M(\boldsymbol \theta)  \mathbf x^i_2 }$.
Taking the derivatives with respect to $\theta^i$ and $\bar \theta^i$, and  then making $\theta^i=\theta^j=\theta$ and $\bar \theta^i=\bar \theta^j=\bar \theta$, leads to:
\begin{equation}\label{eq:nabla1}
\begin{cases}
 \nabla_\theta J_{\mathbf x_1}(\boldsymbol \theta)= 2\Big((1-\frac{1}{n}) (\mathbf R^{1,1} +\gamma {\mathbf B^{1,1}}^\intercal \mathbf M^{1,1}(\boldsymbol \theta) \mathbf B^{1,1})\\
 + \frac{1}{n}( \mathbf R^{2,1} +\gamma {\mathbf B^{2,1}}^\intercal  \mathbf M^{2,1}(\boldsymbol \theta) \mathbf B^{2,1})\Big)\theta \Exp{\Delta x_1 \Delta x_1^\intercal} \\
+2\Big( (1-\frac{1}{n})(\mathbf R^{1,2} +\gamma {\mathbf B^{1,2}}^\intercal  \mathbf M^{1,2}(\boldsymbol \theta) \mathbf B^{1,2}) + \frac{1}{n} (\mathbf R^{2,2}\\
 +\gamma {\mathbf B^{2,2}}^\intercal  \mathbf M^{2,2}(\boldsymbol \theta) \mathbf B^{2,2}) \Big)\bar \theta \Exp{\bar x_1 \Delta x_1^\intercal}+ \gamma \nabla_\theta J_{ \mathbf x_2}(\boldsymbol \theta),\\
 \nabla_{\bar \theta} J_{\mathbf x_1}(\boldsymbol \theta)= 2\Big((1-\frac{1}{n}) (\mathbf R^{1,1} +\gamma {\mathbf B^{1,1}}^\intercal \mathbf M^{1,1}(\boldsymbol \theta) \mathbf B^{1,1})\\
 + \frac{1}{n}( \mathbf R^{2,1} +\gamma {\mathbf B^{2,1}}^\intercal  \mathbf M^{2,1}(\boldsymbol \theta) \mathbf B^{2,1})\Big) \theta \Exp{\Delta x_1 \bar x_1^\intercal} \\
+2\Big( (1-\frac{1}{n})(\mathbf R^{1,2} +\gamma {\mathbf B^{1,2}}^\intercal  \mathbf M^{1,2}(\boldsymbol \theta) \mathbf B^{1,2}) \\
+ \frac{1}{n} (\mathbf R^{2,2}  \hspace{-.1cm}+ \hspace{-.1cm}\gamma {\mathbf B^{2,2}}^\intercal  \mathbf M^{2,2}(\boldsymbol \theta) \mathbf B^{2,2}) \Big) \bar \theta \Exp{\bar x_1 \bar x_1^\intercal}\hspace{-.1cm}+\hspace{-.1cm} \gamma \nabla_{\bar \theta} J_{\mathbf x_2}(\boldsymbol \theta). 
\end{cases}
\end{equation}
The  the rest of the proof follows from the recursive application of~\eqref{eq:nabla1} and equations~\eqref{eq:riccati-bar-m},~\eqref{eq:breve-f} and~\eqref{eq:gradient_matrices}.
\section{Proof of Lemma~\ref{lemma:GD}}\label{sec:proof_lemma:GD}
Let $\tilde{\mathbf P}_n:=\DIAG((1-\frac{1}{n}) \mathbf{I}_{d_u \times d_u}, \frac{1}{n} \mathbf{I}_{d_u \times d_u})$.  We express~\eqref{eq:gradient_1} in terms of the  square matrix $\tilde{\mathbf P}_n$ such that $\Delta J_\theta(\boldsymbol \theta)=:\Delta J^1_\theta(\boldsymbol \theta)+\Delta J^2_\theta(\boldsymbol \theta) $ and $\Delta J_{\bar \theta}(\boldsymbol \theta)=: \Delta J^1_{\bar{\theta}}(\boldsymbol \theta)+\Delta J^2_{\bar{\theta}}(\boldsymbol \theta) $, where
\begin{equation}
\nabla_{\boldsymbol \theta} \tilde J:=
\begin{bmatrix}
\Delta J^1_\theta(\boldsymbol \theta) & \Delta J_{\bar \theta}(\boldsymbol \theta)\\
\Delta J^2_\theta(\boldsymbol \theta) & \Delta J_{\bar \theta}(\boldsymbol \theta)
\end{bmatrix}
=
2\tilde{\mathbf P}_n \mathbf E_{\boldsymbol \theta} \boldsymbol \Sigma_{\boldsymbol \theta}. 
\end{equation}
Following~\cite[Lemma 10]{fazel2018global} and after some algebraic manipulations,   we can derive the following inequality for  sequences   $\{\mathbf{x}^\ast_t\}_{t=1}^\infty$ and $\{\mathbf{u}^\ast_t\}_{t=1}^\infty$ generated by the Nash policy $\boldsymbol \theta^*$. In particular, from~\eqref{eq:gradient_1},   $\Sigma_{\boldsymbol \theta}\succcurlyeq \Exp{(\mathbf x_1 \mathbf x_1^\intercal)}$, and the fact that $\tilde{\mathbf P}_n$ is positive definite for any finite~$n$, it  results that:
$
J(\boldsymbol \theta)- J(\boldsymbol \theta^*) \leq (1-\gamma)\mathbb{E}\sum_{t=1}^\infty \gamma^{t-1} \TR({\mathbf{x}^\ast_t}{\mathbf{x}^\ast_t}^\intercal \mathbf{E}_{\boldsymbol \theta}^\intercal(\mathbf R+\gamma\mathbf B^{\intercal}\mathbf M_{\boldsymbol \theta}\mathbf B)^{-1}\mathbf{E}_{\boldsymbol \theta})=\TR(\Sigma_{\boldsymbol \theta^*}\mathbf{E}_{\boldsymbol \theta}^\intercal(\mathbf R+\gamma\mathbf B^{\intercal}\mathbf M_{\boldsymbol \theta}\mathbf B)^{-1}\mathbf{E}_{\boldsymbol \theta})
\leq
    \frac{\norm{\boldsymbol \Sigma_{\boldsymbol \theta^*}}}{\sigma_{\text{min}}(\mathbf{R})}\TR(\mathbf{E}_{\boldsymbol \theta}^\intercal\mathbf{E}_{\boldsymbol \theta})
    = \frac{\norm{\boldsymbol \Sigma_{\boldsymbol \theta^*}}}{4\sigma_{\text{min}}(\mathbf{R})}\TR({\boldsymbol \Sigma_{\boldsymbol \theta}}^{-1} \nabla_{\boldsymbol \theta} \tilde J^\intercal   \mathbf{P}_n^{-2}  \nabla_{\boldsymbol \theta} \tilde J {\boldsymbol \Sigma_{\boldsymbol \theta}}^{-1})    
	\leq$ $ 
	  \frac{0.25\mu^{-2}\norm{\boldsymbol \Sigma_{\boldsymbol \theta^*}}}{\sigma_{\text{min}}(	\tilde{\mathbf P_n})^2\sigma_{\text{min}}(\mathbf{R})}$  $ \NF{[\nabla_{\theta} J(\boldsymbol \theta), \nabla_{ \bar \theta} J( \boldsymbol \theta)]}^2.
$
	For $n=\infty$, $\tilde{\mathbf P}_n$ is not invertible; however, equation~\eqref{eq:riccati-bar-m} under Assumption~\ref{ass:decoupled} decomposes into two \emph{decoupled} standard Riccati equations with matrices $(A,B,Q,R)$ and $(A,B,Q+S^x,R+S^u)$.  By following the approach proposed in~\cite[Lemma 11]{fazel2018global},  it is straightforward to show  that the cost difference in this case is upper bounded by:   $\frac{\norm{\boldsymbol \Sigma^{1,1}_{\boldsymbol \theta^*}}}{4\sigma_{min}(\text{cov}(x_1))^2\sigma_{\text{min}}(R)} \|\nabla_{ \theta} J(\boldsymbol \theta)\|_F^2$ + $\frac{\norm{\boldsymbol \Sigma^{2,2}_{\boldsymbol \theta^*}}}{4\sigma_{min}(\mathbb{E}[x_1]\mathbb{E}[x_1]^\intercal)^2\sigma_{\text{min}}(R+S^u)} \|\nabla_{\bar  \theta} J(\boldsymbol \theta)\|_F^2$.

\end{document}